\documentclass{article}

\usepackage[english]{babel}
\usepackage{amsmath}
\usepackage{amsthm}
\usepackage{mathtools}
\usepackage{amssymb}
\usepackage{comment}
\usepackage{relsize}

\usepackage[letterpaper,top=2cm,bottom=2cm,left=3cm,right=3cm,marginparwidth=1.75cm]{geometry}

\usepackage{graphicx}
\newtheorem{theorem}{Theorem}
\newtheorem{lemma}{Lemma}
\newtheorem{definition}{Definition}
\usepackage[colorlinks=true, allcolors=blue]{hyperref}

\newcommand{\R}{\mathbb{R}}
\newcommand{\Rn}[1][n]{\mathbb{R}^{#1}}

\newcommand{\K}{\mathcal{K}}
\renewcommand{\L}{{\mathcal{L}}}

\newcommand{\M}{M^n}
\renewcommand{\S}{S^n}

\newcommand{\Co}{C^n}

\renewcommand{\P}{P^n}

\newcommand{\Core}{\mathcal{C}}

\newcommand{\Kc}{{\mathcal{K}^c}}

\newcommand{\valfa}{{v_\alpha}}
\newcommand{\vx}{{v_x}}
\newcommand{\vsigmax}{{v_{\sigma,x}}}

\title{Player-centered incomplete cooperative games}
\author{Martin \v{C}ern\'y, Michel Grabisch}

\begin{document}

\maketitle

\begin{abstract}
    The computation of a solution concept of a cooperative game usually employs values of all coalitions. However, in some applications, the values of some of the coalitions might be unknown due to high costs associated with their determination or simply because it is not possible to determine them exactly. We introduce a method to approximate standard solution concepts based only on partial characteristic function of the cooperative game. In this paper, we build on our previous results and generalise the results of our methods to a significantly larger class of structures of incomplete information.
    
\end{abstract}
	
	%% \linenumbers
	
	%% main text
	\section{Introduction}
	Cooperative game theory, as introduced by von Neumann and Morgenstern
	\cite{Neumann1947} in order to model the strength and power of coalitions, has been
	largely developed in the last fifty years (see, e.g., the monograph of Peleg and
	Sudh\"olter \cite{Peleg2007}), especially the branch of cooperative games with
	transferable utility (the so-called TU-games), where players can transfer freely
	money/utility inside a given coalition. Applications of cooperative games are
	numerous, mainly in economics (fair division of a benefit/cost, bankruptcy
	problems, etc.), social choice (voting games), but also more recently in
	artificial intelligence and machine learning (e.g., the SHAP method
	\cite{Lindberg2017}, based on the well-known Shapley value \cite{Shapley1953}).
	
	Mathematically speaking, TU-games are set functions on a finite universe,
	vanishing on the empty set. They are therefore closely related to discrete
	mathematics in general, in particular combinatorial optimization (see, e.g., the
	monograph of Fujishige \cite{Fujishige2005}), hypergraphs, Boolean and pseudo-Boolean
	functions (see Crama and Hammer \cite{Crama2011}), and some notions of
	combinatorics like the M\"obius function \cite{Rota1964}.
	
	\medskip
	
	Considering a TU-game $v$ on a set of $n$ players $N$, its definition requires
	to know $v(S)$ for any coalition $S\subseteq N$. Usually, in an economic
	context $v(S)$ is interpreted as the maximal benefit the coalition $S$ can
	achieve without the help of the players in $N\setminus S$, and is called the
	{\it worth} of $S$. In a voting context, $v(S)$ is perceived as the power of
	coalition $S$ in an election process. As $n$ grows, it becomes more and more
	unlikely that in practice, $v(S)$ can be known for {\it every} coalition $S$.
	{\it Incomplete games}, which is the topic of the present paper, precisely deals
	with this situation. An incomplete game $v$ is known only on a subcollection
	$\K$ of $2^N$. In reality, $v(S)$ for $S\not\in\K$ is defined but remains
	unknown (this is in opposition with the so-called {\it games with restricted
		cooperation} (see, e.g., \cite{Faigle1989}), where coalitions outside a given
	subcollection $\K$ are ``forbidden'', and $v$ is not defined on these
	coalitions). 
		
	Incomplete cooperative games were introduced by Willson~\cite{Willson1993} in 1993. After introducing the concept, he focused on a value of incomplete games based on the definition of the Shapley value. More than two decades later, Inuiguchi and Masuya renewed the research~\cite{Masuya2016}, focusing on superadditivity of extensions, with Masuya continuing in the research of the Shapley value~\cite{Masuya2021}. Further, Bok et al.~\cite{Bok2020} studied convexity and positivity, Bok and Cerny~\cite{Bok2021} studied 1-convexity, Yu~\cite{Xiaohui2021} extended incomplete games to games with coalition structure, focusing on proportional Owen value, and \v{C}ern\'{y}~\cite{Cerny2022} was dealing with approximations of solution concepts such as the core, the (pre-)kernel and others on \emph{minimal incomplete games}.

	\medskip
	
	There are two natural questions with incomplete games. The first question is:
	What could be the unknown values $v(S)$, $S\not\in\K$? A particular assignment
	for these values defines an {\it extension} of the incomplete game.
	The question makes sense
	only in the case where one imposes some restriction on the game $v$, otherwise,
	any value for $v(S)$ can be taken. Common restrictions are: the game should be
	monotonic (i.e., nondecreasing with respect to set inclusion), superadditive or
	subadditive, convex or concave
	(a.k.a. super/submodular), positive (i.e., with nonnegative M\"obius transform),
	etc. Each of these restrictions corresponds to classical and useful families of
	games. Monotonic games correspond to capacities \cite{Choquet1953}, and are central in
	decision theory. Superadditivity is considered as a basic requirement for a
	TU-game $v$ when $v(S)$ represents a benefit. Convex games are special
	superadditive games, having remarkable properties, both in game theory and in
	decision theory. Lastly, positive games are special convex games which form the
	basis of Dempster-Shafer theory \cite{Shafer1976} for modelling uncertainty. 
	
	The second question is: Since usual solution concepts of TU-games (the core, the
	Shapley value, etc.) cannot be calculated on incomplete games, considering all
	possible extensions in a given family of games, what would be these solutions
	for each extension? If a solution is set-valued like the core, we may consider
	the intersection over all extensions, or the union of the solutions. If a solution is
	point-valued like the Shapley value, we may obtain an interval when considering
	all extensions.
	
	The paper addresses both questions in full detail for a specific family of incomplete
	games, which we call {\it player-centered}. In a player-centered incomplete
	game, the subcollection $\K$ is the set of all coalitions containing a given
	player, say $i$. Such games represent the situation where the information on
	$v(S)$ is gathered by player $i$, who ignores the worth of coalitions where he
	is not present. Put differently, the game is viewed through the eyes of player
	$i$. This kind of situation arises every time players reveal their strategy only
	to members of the coalition they belong to, as it could be the case for example
	if players are competing firms, or more generally in any competitive
	context.
	
	\medskip
	
	The paper is organized as follows. Section 2 presents the basic notions of
	classical TU-games, solution concepts, and introduces incomplete games. Section 3
	presents the player-centered incomplete games and studies their
	extensions. Three types of extensions are considered: positive extensions (with
	positive games), convex and superadditive extensions, and monotone
	extensions. In each case, the geometric structure of the set of extensions is
	studied (extreme points, extreme rays). In Section 4 we study how the main
	solutions concepts (the core, the Shapley value, the $\tau$-value) can be
	approximated, depending on the kind of extension which is considered.  
	
	\section{Preliminaries}
        Throughout the paper, $N$ is a finite set of $n$ elements (the set of
        players). The set of permutations on $N$ is denoted by $\Sigma_n$.  We
        denote by $2^N$ the power set of $N$. The \emph{lattice closure}
        $LC(\K)$ of $\K \subseteq 2^N$ in the partially ordered set
        $(2^N,\subseteq)$ is the inclusion-minimal subset of $2^N$ that contains
        $\K$ and that is closed under the operations of union and intersection of
        sets. To denote the bounds of a real closed interval $a = \left[\underline{a},\overline{a}\right]$, we use the lower and upper bar, respectively.
	
	\begin{definition}\cite{Soltan2015}\label{def:extreme-point}
		Let $P \subseteq \R^n$ be a polyhedron. We say that $e \in P$ is an \emph{extreme point} (of $P$) if for every $x \in \R^n$, we have $(e+x) \in P \wedge (e-x) \in P \implies x = 0$.
	\end{definition}
	
	\begin{theorem}\label{thm:extreme-point-char}
		A nonzero element $x$ of a polyhedron $P \subseteq \R^n$ is an extreme point if and only if there are $n$ linearly independent constraints binding at $x$.
	\end{theorem}

	\subsection{Classical cooperative games}
	For more on cooperative games outside the scope of this text, see, e.g.,
        \cite{Grabisch2016,Peleg2007,Peters2008}. 
	\begin{definition}
		A \emph{cooperative game} is an ordered pair $(N,v)$, where $N$
                is the set of players and $v\colon 2^N \to
		\mathbb{R}$ is its characteristic function. Further, $v(\emptyset) = 0$.
	\end{definition}
	
	We denote the set of $n$-person cooperative games by $\Gamma^n$. Sets $S
        \subseteq N$ are called \emph{coalitions}, while $v(S)$ is called the {\it
          worth} of $S$. We often replace $\left\{i\right\}$ with $i$. To denote the size of coalitions e.g. $N,S,T$, we use $n,s,t$, respectively. We often write $v$ instead of $(N,v)$ whenever there
	is no confusion over the player set and associate the characteristic functions $v\colon 2^N\to\mathbb{R}$ with vectors $v\in\mathbb{R}^{2^{\lvert N \rvert}-1}$. For $x \in \Rn$ and $S \subseteq N$, $x(S) \coloneqq \sum_{i \in S} x_i$. 
	\begin{definition}\label{def:classes-of-games}
		Let $(N,v)$ be a cooperative game. The game is
		\begin{enumerate}
			\item {\it monotone} if it satisfies $v(S) \le v(T)$, \hspace{11ex} $S\subseteq T \subseteq N$;
			\item {\it superadditive} if $v(S) + v(T) \le v(T\cup S)$, \hspace{7ex} $S, T \subseteq N, S\cap T=\emptyset$;
			\item {\it convex} if $v(S) + v(T) \le v(T\cup S) + v(T\cap S)$, \hspace{1,5ex} $S, T \subseteq N$.
		\end{enumerate}
		We denote the set of all monotonic, superadditive and convex $n$-person games by $\M$, $\S$ and $\Co$, respectively.
	\end{definition}
	
	The \emph{Möbius transform} of $(N,v)$ is defined for every $T \subseteq N$ by
	\[
	m^v(T) \coloneqq \sum_{S \subseteq T}(-1)^{\lvert T \rvert - \lvert S \rvert}v(S).
	\]
	\begin{definition}
		A cooperative game $(N,v)$ is \emph{positive} if $m^v(T) \geq 0$ for every $T \subseteq N$. We denote the set of all positive cooperative $n$-person games by $\P$.
	\end{definition}
	Shapley~\cite{Shapley1953} showed that every $v \in \Gamma^n$ can be uniquely expressed as a linear combination $v= \sum_{\emptyset \neq T \subseteq N}m^v(T)u_T$, where $(N,u_T)$ are \emph{unanimity games} defined as 
	\[
	u_T(S) \coloneqq 
	\begin{cases}
		1 & \text{if }T \subseteq S,\\
		0 & \text{otherwise.}\\
	\end{cases}
	\]
	Equivalently, $v(S) = \sum_{\emptyset \neq T \subseteq S}m^v(T)$. Unanimity games are themselves positive games and positive games are convex. Also, a cooperative game $(N,v)$ is convex if and only if for every $A,B \subseteq N, A \subseteq B$ such that $\lvert A \rvert = 2$, it satisfies
	\begin{equation}\label{eq:convex-mobius-char}
		\sum_{T \in [A:B]}m^v(T) \geq 0
	\end{equation}
	where $[A:B] = \{ T \subseteq N \mid A \subseteq T \subseteq B\}$.
	
	\subsection{Solution concepts and payoff vectors}
	The main task of cooperative game theory is to distribute $v(N)$ among
        the players. \emph{Payoff vectors} are vectors $x \in \Rn$, where $x_i$
        represents payoff of player $i$, $\mathbb{X}(v) \coloneqq \{x \in
        \Rn \mid x(N) = v(N)\}$ is the set of \emph{preimputations} and $\mathbb{I}(v) \coloneqq \{ x \in \mathbb{X}(v) \mid x_i \geq v(i) \text{ for every } i \in N\}$ is the set of \emph{imputations}. Special payoff vectors are the \emph{upper vector} $b^v \in \mathbb{R}^n$ defined by $b^v_i\coloneqq  v(N) - v(N\setminus i)$ and the \emph{lower vector} $a^v$ defined as $a^v_i = \max_{S: i \in S} v(S) - b^v(S)$. We also make use of the \emph{gap function} $g^v \colon 2^N \to \R$, defined as $g^v(S) \coloneqq b^v(S) - v(S)$.	
	\begin{definition}\label{def:core}
		The \emph{core} $\Core(v)$ of a cooperative game $(N,v)$ is defined as
		\[\Core(v) \coloneqq \{x \in \mathbb{X}(v) \mid x(S) \geq v(S)
                \text{ for every } S \subseteq N\}.\]
	\end{definition}
	The property $x(S) \geq v(S)$ for every $S \subseteq N$ is called \emph{coalitional rationality}. The core of a superadditive game might be empty, but it is always nonempty for convex games.
	
	\begin{definition}\label{def:shapley}
		The \emph{Shapley value} $\phi \colon \Gamma^n \to \mathbb{R}^n$ is defined as
		\[\phi_i(v) \coloneqq \sum\limits_{S\subseteq N \setminus i}\frac{s!(n-s-1)!}{n!}(v(S\cup i) - v(S)).\]
	\end{definition}
	We often denote $\gamma_S = \frac{s!(n-s-1)!}{n!}$. The Shapley value is \emph{linear}, meaning
	\begin{equation}\label{eq:shapley-linearity}
		\phi(\alpha v+\beta w) = \alpha\phi(v) + \beta\phi(w),\quad
                \alpha,\beta\in \mathbb{R},v,w\in\Gamma^n.
	\end{equation}
	In terms of Möbius transform, it can be expressed as
	\begin{equation}
		\phi_i(v) = \sum_{S \subseteq N, i \in S}\frac{m^v(S)}{\lvert S \rvert}.
	\end{equation}
	\begin{definition}
		The $\tau$-value $\tau(v)$ of a cooperative game $(N,v)$ is the unique convex combination of $a^v$ and $b^v$ satisfying $\sum_{i \in N}\tau_i(v)=v(N)$.
	\end{definition}
	The $\tau$-value does not exist for a cooperative game in general, however, for convex games, it always exists and can by expressed by an explicit formula.
	\begin{theorem}\cite{Driessen1985-thesis}\label{thm:convex-tau}
		For a convex $(N,v)$, the $\tau$-value can be expressed as
		\begin{equation*}
			\tau_i(v) = b^v_i - \frac{g^v(N)}{\sum_{i \in N}g^v(i)}g^v(i).
		\end{equation*}
	\end{theorem}
	In terms of Möbius transform, this can be rewritten as
	\begin{equation}\label{eq:tau-value-mobious}
		\tau_i(v) = m^v(i) + \frac{\sum_{S \subseteq N, \lvert S \rvert > 1}m^v(S)}{\sum_{S \subseteq N, \lvert S \rvert > 1}s\cdot m^v(S)}\sum_{S \subseteq N, \lvert S \rvert > 1, i \in S}m^v(S).
	\end{equation}
	\subsection{Incomplete cooperative games}
	\begin{definition}\emph{(Incomplete game)}
		An incomplete game is a tuple $(N,\K,v)$ where $N$ is the set of
                players, $\K \subseteq 2^N$ is the set of coalitions with known values and $v\colon \K \to \R$ is the characteristic function. Further, $\emptyset \in \K$ and $v(\emptyset)=0$.
	\end{definition}
	
	We denote the set of $n$-person incomplete games with $\K$ by $\Gamma^n(\K)$. Further, $\Kc \coloneqq 2^N \setminus \K$. Incomplete games represent partial information about cooperation of players. This means that the formation of $S \in \Kc$ is possible even though we do not know its value. This distinguishes the approach from the theory of \emph{games with restricted cooperation}, where the formation of $S \in \Kc$ is impossible. For more on restricted games, see~\cite{Grabisch2016}. 
	
	Since the formation of $S \in \Kc$ is possible in our approach, we would
        like to consider $S$ when determining the payoff distribution. One way
        to do this is to impose bounds on the worth of $S$, based on the properties of the game. Imagine our knowledge about an underlying complete game $(N,v)$ is represented by $(N,\K,v)$ and we also know that $(N,v)$ satisfies further properties, e.g., convexity. Based on this, we can consider the set of possible candidates for $(N,v)$. The idea is formally captured in the following definition.
	
	\begin{definition}\emph{($C$-extension)}\label{def:extension}
		Let $C\subseteq \Gamma^n$ be a class of $n$-person cooperative
                games. The game $(N,w) \in C$ is a \emph{$C$-extension} of an incomplete game $(N,\K,v)$ if $w(S)=v(S)$ for every $S \in \mathcal K$. 
	\end{definition}
	
	By $C(\K,v)$ or simply $C(v)$ we denote the set of $C$-extensions of an incomplete game
        $(N,\K,v)$. We write \emph{$C(v)$-extension} when we want to emphasize
        $(N,\K,v)$. We say that $(N,\K,v)$ is \emph{$C$-extendable} if $C(v)$ is nonempty. The set of all $C$-extendable incomplete games with fixed $\K$ is denoted by $C(\K)$.
	
	From the inclusion relations between $\M,\S,\Co$ and $\P$, it holds $\M(v) \supseteq \P(v)$ and $\S(v) \supseteq \Co(v) \supseteq \P(v)$. Further, sets $\M(v)$, $\S(v)$, $\Co(v)$ and $\P(v)$ are described by systems of linear equalities and inequalities, therefore they form polyhedra. Polyhedra can be also described by their extreme points and extreme rays (see~\cite{Soltan2015} on fundamentals of convex sets). We focus on finding this alternative description in the next section. In the rest of this subsection, we collect previous results on $\Co$-extendability and properties of the sets of $\P$-extensions. Following is a modification of Theorem 7 from~\cite{Bhaskar2018}.
	
	\begin{theorem}\cite{Bhaskar2018} \label{thm:BK2018}
		Let $(N,\K,v)$ be an incomplete cooperative game and 
		\[\mathcal{F} \coloneqq LC(\K) \cap \{S \subseteq N \mid
                \underline{S},\overline{S} \in\K \text{ s.t. } \underline{S}
                \subseteq S \subseteq \overline{S}\}.\] Then $(N,\K,v)$ is $\Co$-extendable if and only if there is convex $w\colon \mathcal{F} \to \mathbb{R}$ such that $w(S)=v(S)$ for $S \in \K$.
	\end{theorem}
	
	The last two results concern boundedness of the set of $\P$-extensions
        and a characterisation of extreme games of $\P(v)$. For a game $(N,w)$, denote by $E(w)\coloneqq\{T \subseteq N \mid m^w(T)=0\}$ the set of coalitions with zero Möbius transform. 
	\begin{theorem}\label{thm:vertex-positive}\cite{Bok2020}
		Let  $(N,\K,v)$ be a $\P$-extendable incomplete game. Then $(N,e)$ is an extreme game of the set of $\P$-extensions if and only if there is no $\P$-extension $(N,w)$ such that $E(e) \subsetneq E(w)$.
	\end{theorem}
	\begin{theorem}\label{prop:positive-bound}\cite{Bok2020}
		Let $(N,\K,v)$ be a $\P$-extendable incomplete game. The set of positive extensions $\P(v)$ is bounded if and only if $N \in \K$.
	\end{theorem}

	\section{Player-centered incomplete games}
	In this section, we restrict to player-centered incomplete games. We say
        that an incomplete game $(N,\K,v)$ is {\it player centered} if $\K$ has
        the form $\K = \{S \subseteq N \mid i \in S\}\cup\{\emptyset\}$ for some
        player $i$. To stress the role of player $i$, we may refer to the game as
        $i$-centered. Unless otherwise specified, all
          player-centered games in the following will be considered as $i$-centered.

        The structure of $\K$ allows player-centered incomplete
        games to inherit properties of their complete counterparts. Namely,
        since $\K$ is closed under finite intersections and unions, we call
        $(N,\K,v)$ \textit{convex} if $v(S) + v(T) \leq v(S \cap T) + v(S \cup
        T)$ for every $S,T \in \K$. In a similar manner, one can define
        \textit{superadditivity} and \textit{monotonicity} of $(N,\K,v)$. The
        closedness of $\K$ on subsets allows to define the \textit{Möbius
          transform} $m^v$ of a player-centered game $(N,\K,v)$ as $m^v(S)
        \coloneqq\sum_{T \subseteq S, T\in \K}(-1)^{\lvert S \setminus T
          \rvert}v(T)$ for any $S\in\K$. We call $(N,\K,v)$ \textit{positive}, if $m^v(S) \geq 0$ for all $S \in \K$.
	
	We investigate different sets of extensions, namely positive, convex,
        superadditive and monotonic ones. The dimension of any of the sets of
        $C$-extensions that we study is bounded from above by $\lvert \Kc \rvert
        = 2^{n-1}-1$. This is because any extension can be represented as a
        vector $w \in \R^{2^n-1}$ where $2^{n-1}$ values are
        fixed. Equivalently, we can view the game as the vector of its Möbius
        transform $m^w \in 2^{n}-1$, where again, the Möbius transform of $2^{n-1}$ coalitions has to be fixed. This view will be important when applying Theorem~\ref{thm:extreme-point-char} to extreme points of sets.
	
		We will see in this section that if an incomplete game lies in a class
        $C$, which is one of the mentioned classes, it is also
        $C$-extendable.	Throughout this section, we employ extensions $(N,v_0)$ and $(N,v_1)$ of $(N,\K,v)$, which are defined as
	\begin{equation}\label{eq:min-max-vertices}
		v_0(S) \coloneqq \begin{cases}
			v(S) & \text{if } S \in \K,\\
			v(S \cup i) & \text{if } S \in \Kc,\\
		\end{cases}
		\hspace{0.2in} \text{and} \hspace{0.2in}
		v_1(S) \coloneqq \begin{cases}
			v(S) & \text{if } S \in \K,\\
			0 & \text{if } S \in \Kc.\\
		\end{cases}
	\end{equation}
	or in terms of Möbius transform, we have $m^{v_0}(i)=m^{v_1}(i)=v(i)$ and for $S \neq \{i\}$, we have
	\begin{equation}
		m^{v_0}(S) \coloneqq \begin{cases}
			0& \text{if } S \in \K,\\
			m^{v}(S\cup i) & \text{if } S \notin \K,\\
		\end{cases}
		\hspace{0.4in}
		m^{v_1}(S) \coloneqq \begin{cases}
			m^v(S) & \text{if } S \in \K,\\
			0 & \text{if } S \notin \K.\\
		\end{cases}
	\end{equation}
	
	These extensions reflect player $i$'s marginal contribution to the game. For every coalition $S \subseteq N \setminus i$, $v_0(S \cup i) - v_0(S) = 0$ and $v_1(S \cup i) - v_1(S) = v(S\cup i)$. Thus, in case of $(N,v_0)$ player $i$ does not contribute to the game as opposed to $(N,v_1)$, where his contribution to $S \cup i$ is in a sense maximal because it is equal to the value of the coalition.
	
	\subsection{Positive extensions}
	In this section, we show that $(N,\K,v)$ is $\P$-extendable if and only
        if it is positive, and we describe the extreme points of the set of $\P$-extensions.
	
	\begin{theorem}
		A player-centered incomplete game $(N,\K,v)$ is $\P$-extendable if and only if it is positive.
	\end{theorem}
	
	\begin{proof}
		If $(N,\K,v)$ is positive, so is extension $(N,v_1)$ defined in~\eqref{eq:min-max-vertices}, thus $(N,\K,v)$ is $\P$-extendable. If $(N,\K,v)$ is $\P$-extendable, let $(N,p)$ be its $\P(v)$-extension. By induction on $\lvert S \rvert$, we show for $S \in \K$ that $m^v(S) = m^p(S) + m^p(S \setminus i) \geq 0$. For $S = \{i\}$, we have
		\[
		v(i) = m^v(i) = m^p(i) + m^p(\emptyset) \geq 0.
		\]
		Further fix $S \in \K$ and suppose for all its subsets $T \in \K, T \subsetneq S$ that it holds $m^v(T) = m^p(T) + m^p(T \setminus i)\geq 0$. From $v(S)=p(S)$, it follows $\sum_{T \subseteq S, T \in \K}m^v(T) = \sum_{T \subseteq S}m^p(T)$ which equals
		\[
		\sum_{T \subsetneq S, T \in \K}\left(m^p(T) + m^p(T \setminus i)\right) + m^p(S) + m^p(S \setminus i).
		\]
		By induction hypothesis,  $m^v(S) = m^p(S) + m^p(S \setminus i) \geq 0$ follows.
	\end{proof}
	
	For every $\P$-extension $(N,p)$, it holds $m^v(S \cup i) = m^p(S \cup i) + m^p(S) \geq 0$ for every $S \in \Kc$. Following are $\P$-extensions for which exactly one of the values $m^p(S \cup i)$, $m^p(S)$, for every $S \in \Kc$, is non-zero. For $x \in \{0,1\}^{\Kc}$, we define $(N,v_x)$ via its Möbius transform for every $S \in \Kc$ as
	\begin{equation}\label{eq:def-vx}
		\hspace{-0.04in}m^{v_x}(S)\coloneqq\begin{cases}
			m^v(S \cup i) & \text{if } x_S=0,\\
			0 & \text{if } x_S =1,\\
		\end{cases}\hspace{0.15in}
		m^{v_x}(S\cup i ) \coloneqq \begin{cases}
			0 & \text{if } x_S=0,\\
			m^v(S \cup i) & \text{if } x_S=1,\\
		\end{cases}\hspace{-0.15in}
	\end{equation}
	with $m^{v_x}(i)\coloneqq v(i)$ and $m^{v_x}(\emptyset)=0$. Each entry $x_S$ of the vector $x$ distinguishes which one of coalitions $S$ and $S \cup i$ has surplus equal to $m^v(S \cup i)$ and which one is equal to $0$. In other words, we decide if the surplus of the coalition is due to player $i$ or if he is useless for this coalition. Notice, $v_0 = v_x$ where $x=0$ and $v_1=v_x$ where $x=1$. In the next theorem, we show that these games form the vertices of the set of $\P$-extensions.
	\begin{theorem}
		For a $\P$-extendable $i$-centered incomplete game $(N,\K,v)$, games $(N,v_x)$ defined in~\eqref{eq:def-vx} are the only extreme points of the set of $\P$-extensions.
	\end{theorem}
	\begin{proof}
		For $S \in \K$, we have
		\[
		v_x(S) = \sum_{T \subseteq S,}m^{v_x}(T) = m^{v}(i) + \sum_{T \subseteq S, T \in \Kc} \left[m^{v_x}(T) + m^{v_x}(T\cup i)\right].
		\]
		As $m^{v_x}(T) + m^{v_x}(T \cup i)=m^v(T \cup i)$ for every $S \in \Kc$, it holds
		\[
		v_x(S) = m^v(i) + \sum_{T \subseteq S, S \in \Kc}m^v(T \cup i) = \sum_{T \subseteq S, T \in \K}m^v(T) = v(S).
		\]
		therefore $(N,v_x)$ are extensions of $(N,\K,v)$. Further, because $m^{v_x}(T)=0$ for $2^{n-1}-1$ coalitions, by Theorem~\ref{thm:extreme-point-char}, games $(N,v_x)$ are extreme points of the set of $\P$-extensions.
		
		To show there is no other extreme point, we employ Theorem~\ref{thm:vertex-positive}. For a contradiction, suppose there is an extreme point $(N,e)$ which is different from all $(N,v_x)$. This means there is $S \in \Kc$ such that both
		\[
		m^e(S) > 0 \hspace{0.2in}\text{and} \hspace{0.2in}m^e(S \cup i) > 0.
		\]
		We can define $(N,e')$ by its Möbius transform as
		\[
		m^{e'}(T) \coloneqq \begin{cases}
			m^{e}(S) + m^{e}(S\cup i) & \text{if }T = S,\\
			0 & \text{if }T = S \cup i,\\
			m^e(T) &\text{otherwise.}\\
		\end{cases}
		\]
		Game $(N,e')$ is clearly positive. Also, it is an extension of $(N,\K,v)$, because for $T \in \K$,
		\[
		e'(T) = \sum_{T' \subseteq T}m^{e'}(T') = \sum_{T' \subseteq T}m^e(T') = e(T).
		\]
		But because $E(e') \subsetneq E(e)$, by Theorem~\ref{thm:vertex-positive}, $(N,e)$ is not an extreme game.
	\end{proof}
	
	We note that the number of vertices of the set of $\P$-extensions is at most $2^{\lvert\Kc\rvert}$, depending on values of $(N,\K,v)$. Also, if $(N,w)$ is additive, there is only one $\P$-extension.
	
	We present another way to express the set of $\P$-extensions. Let $(N,\valfa)$ be a $\P$-extension defined as $\valfa=\sum_{x \in \{0,1\}^\Kc}\alpha_x v_x$, where $\alpha_x \geq 0$ for every $x$ and $\sum_{x \in \{0,1\}^\Kc}\alpha_x=1$. It can be expressed in terms of Möbius transform for every $S \in \Kc$ as
	\begin{equation}\label{eq:positive-extensions-1}
		m^\valfa(S) = m^v(S\cup i)\sum_{x:x_S=0}\alpha_x,\hspace{0.4in}m^\valfa(S \cup i ) = m^v(S\cup i)\sum_{x:x_S=1}\alpha_x,
	\end{equation}
	$m^\valfa(\emptyset) = 0$, and $m^\valfa(i)=v(i)$.
	We may denote $\beta_S \coloneqq \sum_{x:x_S=0}\alpha_x$, which reduces~\eqref{eq:positive-extensions-1} to
	\[
	m^\valfa(S) = \beta_S m^v(S \cup i)\hspace{0.2in}\text{and}\hspace{0.2in}m^\valfa(S \cup i) = (1-\beta_S) m^v(S \cup i).
	\]
	Therefore, we have
	\[
	\valfa = \sum_{\emptyset \neq S \subseteq N}m^\valfa(S)u_S = m^v(i)u_{\{i\}}+\sum_{S \in \Kc}m^v(S\cup i)\left[\beta_Su_S + (1-\beta_S)u_{S \cup i}\right].
	\]
	We will show that choosing $\beta_S \in \left[0,1\right]$ for every $S \in \Kc$ yields a $\P$-extension.
	\begin{theorem}
		For a $\P$-extendable $i$-centered incomplete game $(N,\K,v)$, every $\P$-extension $(N,p)$ can be expressed as
		\begin{equation}\label{eq:positive-extensions-alternative}
			p=m^v(i)u_{\{i\}}+\sum_{S \in \Kc} m^v(S \cup i)\left[\beta_Su_S + (1-\beta_S)u_{S \cup i}\right] 
		\end{equation}
		where $\beta_S \in \left[0,1\right]$ for every $S \in \Kc$.
	%	\begin{equation}\label{eq:positive-extensions-alternative}
	%	\P(v) = \left\{m^v(i)u_{\{i\}}+\sum_{S \in \Kc} m^v(S \cup i)\left[\beta_Su_S + (1-\beta_S)u_{S \cup i}\right] \mid \beta_S \in \left[0,1\right]\text{ for } S \in \Kc\right\}.
	%\end{equation}
	\end{theorem}
	\begin{proof}
		Denote by $P$ the set of games from~\eqref{eq:positive-extensions-alternative}. We already showed $\P(v) \subseteq P$ and it is immediate that $p \in P$ is positive. It remains to show that $p$ is an extension of $(N,\K,v)$. For $S \in \K$, we have
		\[
		p(S) = \sum_{T \subseteq S}m^p(S) = m^v(i) +\sum_{T \subseteq S, T \in \Kc}\left( m^p(T) + m^p(T\cup i) \right)\]
		which is equal to
		\[m^v(i)+\sum_{T \subseteq S, T \in \Kc}\left[\beta_T m^v(T \cup i) + (1 -\beta_T)m^v(T \cup i)\right]
		\]
		or to
		\[
		m^v(i)+\sum_{T \subseteq S, T \in \Kc}m^v(T \cup i) = v(S).\]
	\end{proof}
	
	\subsection{Convex and superadditive extensions}
	In this section, we show that $(N,\K,v)$ is always $\S$-extendable and we also show that $\Co$-extendability is equivalent to convexity of $(N,\K,v)$. Further, if nonempty, both sets are always unbounded. It even holds for two different $i$-centered games that the sets of extreme rays of the set of $\S$-extensions (as well as $\Co$-extensions) are the same. By restricting to monotonicity of both sets of extensions, both sets become bounded and we obtain a nice hierarchy with the sets of $\P$-extensions and $\M$-extensions. 
	
	\begin{theorem}
		A player-centered incomplete game $(N,\K,v)$ is $\Co$-extendable if and only if it is convex.
	\end{theorem}
	\begin{proof}
		Follows from Theorem~\ref{thm:BK2018}.
	\end{proof}
	\begin{theorem}
		A player-centered incomplete game $(N,\K,v)$ is $\S$-extendable.
	\end{theorem}
	\begin{proof}
		According to Definitions~\ref{def:classes-of-games} and~\ref{def:extension}, $\S(v)$-extension $(N,w)$ has to satisfy tree types of conditions:
		\begin{equation}\label{eq:sextensions1}
			w(S) + w(T) \leq w(S \cup T)\text{ for }S,T \in \Kc\text{ such that }S \cap T = \emptyset,
		\end{equation}
		\begin{equation}\label{eq:sextensions2}
			w(S) \leq v(S \cup T) - v(T)\text{ for }S \in \Kc, T \in \K\setminus \{\emptyset\}\text{ such that }S \cap T = \emptyset,
		\end{equation}
		\begin{equation}\label{eq:sextensions3}
			w(S) = v(S)\text{ for }S \in \K.
		\end{equation}
		Notice that for $S,T \in \K$, no conditions have to be satisfied, because $S\cap T \neq \emptyset$. We can express $\Kc \cup \{\emptyset\} = 2^{N \setminus i}$, thus Conditions~\eqref{eq:sextensions1} are satisfied by any additive game $\beta \in \R^{n-1}$. To satisfy Conditions~\eqref{eq:sextensions2}, we define
		\[
		\beta_k \coloneqq \min_{\substack{S \in \Kc, T \in \K \setminus \{\emptyset\},\\ S \cap T = \emptyset}}\frac{v(S \cup T) - v(T)}{n}
		\]
		and denote by $S^*,T^*$ coalitions for which the minimum is attained. For disjoint $S \in \Kc, T \in \K \setminus \{\emptyset\}$, it holds
		\[\beta(S^*) = \frac{s}{n}\left(v(S^* \cup T^*) - v(T^*)\right) \leq \frac{s}{n}\left(v(S \cup T) - v(S)\right) \leq v(S \cup T) - v(T).\]
		Finally, we define $(N,w)$ as
		\[
		w(S) \coloneqq \begin{cases}
			v(S)&\text{if } S \in \K,\\
			\beta(S)& \text{if } S \in \Kc.\\
		\end{cases}
		\]
		From its construction, it follows $(N,w)$ is $\S(v)$-extension.
	\end{proof}
	
	We proceed with the analysis of the recession cone of the set of
        $\Co$-extensions. Let $\K$ be fixed, $i$-centered. In this section we
        use the notation $C^n(\K,v)$ for the set of
	convex extensions of $(N,\K,v)$. This set is defined by the following set of
	inequalities:
	\begin{align}
		w(S) + w(T) - w(S\cap T) - w(S\cup T) & \leq 0, \quad S,T\subseteq N\label{eq:1}\\
		w(S) & = v(S), \quad S\in \K.
	\end{align}
	As $v$ is given, observe that this implies that the dimension of $C^n(\K,v)$ is
	at most $|\Kc|=2^{n-1}-1$.
	
	The recession cone of $C^n(\K,v)$ is simply $C^n(\K,0)$. Taking advantage that
	$w(S)=0$ for all $S\in\K$, we can project $C^n(\K,0)$ into $\R^{\Kc}$. Then
	(\ref{eq:1}) reduces to
	\begin{align*}
		w(S) + w(T) - w(S\cap T) - w(S\cup T) & \leq 0, \quad S,T\in\Kc\\
		w(S) - w(S\cap T) & \leq 0, \quad S\in\Kc,T\in\K.
	\end{align*}
	As $T\in\K,T\neq\emptyset,\{i\}$ is equivalent to $T=T'\cup \{i\}$ with $T'\in\Kc$,
	and since $T=\emptyset$, $T=\{i\}$ entail $w(S\cap T)=w(\emptyset)=0$, we  get
	\begin{align*}
		w(S) + w(T) - w(S\cap T) - w(S\cup T) & \leq 0, \quad S,T\in\Kc\\
		w(S)-w(S')& \leq  0, \quad S,S'\in \Kc,S'\subseteq S\\
		w(S) & \leq 0, \quad S\in \Kc. 
	\end{align*}
	Getting rid of redundant inequalities, we finally get
	\begin{align}
		w(S\setminus \{j\})+ w(S\setminus \{k\}) -w(S\setminus\{j,k\}) - w(S) & \leq 0,\quad
		S\in\Kc, j,k\in S, |S|> 2\label{eq:2}\\
		w(\{j\}) + w(\{k\}) - w(\{j,k\}) & \leq 0, \quad j,k\neq i\label{eq:3}\\
		w(\{j\}) & \leq 0, \quad j\neq i\label{eq:4}\\
		w(S)-w(S\setminus \{j\}) & \leq 0, \quad S\in \Kc,|S|>1, j\in S.\label{eq:5a}
	\end{align}
	(for the first set of inequalities, see Grabisch and Kroupa~\cite{Grabisch2019})
	
	\begin{lemma}\label{lem:1}
		$C^n(\K,0)$ is pointed (i.e., it contains no line).
	\end{lemma}
	\begin{proof}
		Recall that a cone defined by $Ax\leq 0$ is pointed iff $Ax=0$ has 0 as unique solution. From
		(\ref{eq:4}) we get $w(\{j\})=0$ for all $\{j\}\in\Kc$. Then from (\ref{eq:5a}),
		we obtain $w(\{j,k\})=0$ for all $\{j,k\}\in\Kc$. Reusing (\ref{eq:5a}), we
		finally obtain that $w(S)=0$ for all $S\in\Kc$.
	\end{proof}
	
	We define $(N,e_{S_0})$, $S_0\in\Kc$, as $e_{S_0}(T)\coloneqq0$
	for all $T\in \K$, and for $T\in \Kc$:
	\begin{equation}\label{eq:5}
		m^{e_{S_0}}(T)\coloneqq\begin{cases}
			(-1)^{|T|}, & \text{ if } T\subseteq S_0\\
			0, & \text{ otherwise}.
		\end{cases}
	\end{equation}
	
	\begin{lemma}
		For $S_0 \in \Kc$, $e_{S_0}$ belong to $\Co(\K,0)$.
	\end{lemma}
	\begin{proof}
		The recession cone of the set of $\Co$-extensions of $(N,\K,v)$ can be expressed as
		\[
		\{e \in \R^{2^n-1}\mid Ce \leq 0\}
		\]
		where rows of $Ce \leq 0$ correspond to
		\begin{equation}\label{eq:C-extreme-ray1}
			e(S) + e(T) \leq e(S \cap T) + e(S \cup T)
		\end{equation}
		for $S,T \subseteq N$ and
		\begin{equation}\label{eq:C-extreme-ray2}
			e(S)=0
		\end{equation}
		for $S \in \K$. To prove $(N,e_{S_0})$ is in the recession cone, we have to prove it is convex (which is equivalent to conditions~\eqref{eq:C-extreme-ray1} and conditions~\eqref{eq:C-extreme-ray2}).
		By~\eqref{eq:convex-mobius-char}, the game $(N,e_{S_0})$ is convex if and only if for all $A,B \subseteq N$ satisfying $A \subseteq B$, $\lvert A \rvert = 2$, it holds
		\[
		\sum_{T \in [A:B]}m^{e_{S_0}}(T) \geq 0.
		\]
		First, if $B \not\subseteq S_0 \cup i$, then from~\eqref{eq:5}, it follows that $m^{e_{S_0}}(T)=0$ for every $T \in [A:B]$, thus
		\[
		\sum_{T \in [A:B]}m^{e_{S_0}}(T) = 0.
		\]
		Second, if $B \subseteq S_0 \cup i$, then $m^{e_{S_0}}(T) = (-1)^T$ for every $T \in \left[A:B\right]$. Denote by $C = B \setminus A$. Then
		\[
		\sum_{T \in \left[A:B\right]}m^{e_{S_0}}(T) = \sum_{A \subseteq T \subseteq B}(-1)^{\lvert T \rvert} = \sum_{t=0}^{\lvert C \rvert}(-1)^t\genfrac(){0pt}{2}{\lvert C \rvert}{t} = 0.
		\]
		Game $(N,e_{S_0})$ is therefore convex. Further, for $T \in \Kc$, we have
		\[
		e_{S_0}(T \cup i) = \sum_{X \subseteq T \cup i}m^{e_{S_0}}(X).
		\]
		From the definition of $(N,e_{S_0})$, when $X \not\subseteq S \cup i$, then $m^{e_{S_0}}(X) = 0$, thus for $A=(S \cap T)\cup i$, we have
		\[
		e_S(T \cup i)  = \sum_{X \subseteq A}m^{e_S}(X).
		\]
		This expression is almost equal to counting the difference between the number of even and odd subsets. The difference is that we do not add $1$ for the empty coalition $\emptyset$ and also for $\{i\}$, $m^{e_{S_0}}(i)=0$, therefore we also do not add $-1$, thus $\sum_{X \subseteq A}m^{e_{S_0}}(X) = 0$, which concludes the proof.
	\end{proof}
	
	\begin{lemma}
		For any $S_0\in\Kc$, $e_{S_0}$ is an extreme ray of $\Co(\K,0)$.
	\end{lemma}
	\begin{proof}
		It remains to show that $e_{S_0}$ is extreme. This amounts to
		showing that the set of solutions to $A^=x=0$ is a 1-dim vector space, where $A^=$
		is the set of tight inequalities for $e_{S_0}$.
		
		It is convenient to rewrite the system defining $C^n(\K,0)$ in terms of the
		M\"obius transform of $w$, denoted by $m^w$. We obtain:
		\begin{align}
			-\sum_{\substack{T\subseteq S\\ T\ni j,k}}m^w(T) & \leq 0, \quad S\in \Kc,
			j,k\in S, |S|>2\label{eq:2m}\\
			-m^w(\{j,k\}) & \leq 0, \quad \{j,k\}\in\Kc\label{eq:3m}\\
			m^w(\{j\}) & \leq 0, \quad \{j\}\in\Kc\label{eq:4m}\\
			\sum_{\substack{T\subseteq S\\T\ni j}} m^w(T) & \leq 0, \quad S\in
			\Kc,|S|>1,j\in S.\label{eq:5m}
		\end{align}
		Let us check which inequalities are tight.
		\begin{enumerate}
			\item Ineq. (\ref{eq:2m}): If $\{j,k\}\not\subseteq S_0$, then no $T\subseteq S_0$ can
			contain $j,k$, hence all terms in the sum are zero, and the inequality is
			tight.
			Now, if $\{j,k\}\subseteq S_0$, then
			\[
			\sum_{\substack{T\subseteq S\\T\ni
					j,k}}m^{e_{S_0}}(T)=\sum_{\substack{T\subseteq S\cap S_0\\T\ni j,k}}(-1)^{|T|}=0,
			\]
			except when $S\cap S_0=\{j,k\}$.
			In summary, (\ref{eq:2m}) is always tight, except when $S\cap S_0=\{j,k\}$.
			\item Clearly, (\ref{eq:3m}) is tight iff $\{j,k\}\not\subseteq S_0$.
			\item (\ref{eq:4m}) is tight iff $j\not\in S_0$.
			\item Ineq. (\ref{eq:5m}): If $j\not\in S_0$, then all terms of the sum are zero
			and the inequality is tight. If $j\in S_0$, then we obtain
			\[
			\sum_{\substack{T\subseteq S\\T\ni j}}m^{e_{S_0}}(T) =
			\sum_{\substack{T\subseteq S\cap S_0\\T\ni j}}(-1)^{|T|} = 0.
			\]
			except when $|S\cap S_0|=1$.
			Hence, in all cases except when $|S\cap S_0|=1$, the inequality is tight.
		\end{enumerate}
		
		In summary, the system of tight inequalities is formed by all equalities
		(\ref{eq:2m}) except when $S\cap S_0=\{j,k\}$, all equalities (\ref{eq:5m})
		except when $S\cap S_0=\{j\}$, and the system
		\begin{align}
			m^w(\{j\}) & = 0,\quad j\not\in S_0\label{eq:3m1}\\
			m^w(\{j,k\}) & = 0,\quad \{j,k\}\not\subseteq S_0\label{eq:4m1}.
		\end{align}
		Consider first the case $|S_0|=1$, say, $S_0=\{1\}$. Then by (\ref{eq:3m1}) and
		(\ref{eq:4m1}), $m^w(T)$ is determined (=0) for $|T|=1,2$, except
		$m^w(\{1\})$. Observe that no other equation contains $m^w(\{1\})$, so that it
		remains undetermined. Now, consider $T$ s.t. $|T|=3$. Using (\ref{eq:2m}) with
		$S=T$ yields $m^w(T)=0$. Hence we obtain $m^w(T)=0$ for all $T$ such that
		$|T|=3$. Iterating the process, we finally get that $m^w(T)=0$ for all $T$ such
		that $|T|\geq 2$. As a conclusion, the set of solutions has dimension 1.
		
		\medskip
		
		Consider then that $|S_0|>1$, w.l.o.g., $S_0=\{1,\ldots, p\}$. We show that all
		$m^w(T)$ can be expressed in terms of $m^w(\{1\})$. By (\ref{eq:3m1}) and
		(\ref{eq:4m1}), $m^w(T)$ is determined (=0) for all $T$ s.t. $|T|=1,2$, except
		when $T\subseteq S_0$. Observe that no equation (\ref{eq:2m}) contains
		$m^w(\{1\})$ and no equation (\ref{eq:5m}) contains $m^w(\{1\})$ alone.
		Consider $S=\{1,k\}\subseteq S_0$ and $j=1$. Applying
		(\ref{eq:5m}) yields $m^w(\{1,k\})=-m^w(\{1\})$. Taking now $j=k$ and
		$S=\{1,k\}$ in
		(\ref{eq:5m}) yields $m^w(\{k\})=m^w(\{1\})$. This in turn permits to
		determine all $m^w(\{k,l\})$, $\{k,l\}\subseteq S_0$ by taking $S=\{k,l\}$ and
		$j=k$ in (\ref{eq:5m}). In this way, we can determine all values of $m^w(T)$ for
		pairs and singletons in $S_0$. 
		By repeated use of
		(\ref{eq:5m}), one can then determine all values of $m^w(T)$ with $T\subseteq
		S_0$, proceeding with $|T|=3$ first, then $|T|=4$, etc. 
		
		It remains to consider the case $T\not\subseteq S_0$. Suppose $|T\cap
		S_0|\leq 1$. Starting from $|T|=3$ and letting $S=T$ and choosing arbitrary
		$j,k$ in (\ref{eq:2m}) yields
		$m^w(T)=0$. Repeating the process with step by step higher cardinalities yields
		$m^w(T)=0$ for all such $T$. Suppose now $|T\cap S_0|>1$. Then start with
		$T$ such that $|T\setminus S_0|=1$, and use (\ref{eq:5m}) with $S=T$ and $j\in T\cap S_0$. As $m^w$ is
		determined for all subsets of $S_0$, this determines the value of $m^w(T)$. Then
		$m^w$ is determined for all $T$ such that $|T\setminus S_0|=1$. Repeating the
		process with higher cardinalities of $T\setminus S_0$ determines the value of
		$m^w(T)$ for all $T$.  
		
	\end{proof}
	
	Finally, we turn our attention towards bounding the set $\Co(\K,v)$. The following lemma shows that the set is bounded if and only if the values of the singletons are bounded. 
	
	\begin{lemma}\label{lem:extreme-rays-convex}
		Let $(N,e)$ be an extreme ray of the recession cone $\Co(\K,0)$. Then there is $\{k\} \in \Kc$ such that $e(\{k\}) < 0$.
	\end{lemma}
	\begin{proof}
		Suppose $e(\{k\}) = 0$ for all $\{k\} \in \Kc$. There is $S \in \Kc$ such that $e(S) < 0$, otherwise $(N,e)$ is not an extreme ray. However, from convexity (actually, superadditivity is enough), this leads to a contradiction, because
		\[
		0 = \sum_{k \in S}e(\{k\}) \leq e(S) < 0.
		\]
	\end{proof}

	In a similar manner, we can show that every $i$-centered incomplete game has the same recession cone of $\S$-extensions, denoted by $\S(\K,0)$. It is pointed and nonempty because it contains $(N,-u_k)$ where $u_k$ is a unanimity game. These actually form a subset of extreme rays, however we omit the proof.
	
	\begin{lemma}\label{lem:extreme-rays-super}
		Let $(N,e)$ be an extreme ray of the recession cone $\S(\K,0)$. Then there is $\{k\} \in \Kc$ such that $e(\{k\}) < 0$.
	\end{lemma}
	\begin{proof}
		Follows from the proof of Lemma~\ref{lem:extreme-rays-convex}.
	\end{proof}
	
	For both $\S$-extensions and $\Co$-extensions, games $(N,e_{S_0})$ and $(N,-u_k)$ are not their only extreme rays. In Table 1, we give an overview of the number of extreme rays of sets of $\S$-extensions and $\Co$-extensions. These results were achieved numerically.
	\medskip
	\begin{table}[h]
		\label{table1}
		\caption{Number of extreme rays of $\S(v)$ and $\Co(v)$ in comparison with extreme rays of form $(N,-u_k)$ and $(N,e_S)$.}
		\begin{center}
			\begin{tabular}{ |c|c|c|c|c|c| }
				\hline
				$n$ & $1$ & $2$ & $3$ & $4$ & $5$ \\
				\hline
				$rays(\S)$ & 0 & 1 & 4 & 22 & 3120 \\
				\hline
				$(N,-u_{k})$ & 0 & 1 & 2 & 3 & 4\\
				\hline
				$rays(\Co)$ & 0 & 1 & 3 & 8 & 41\\
				\hline
				$(N,e_{S_0})$ & 0 & 1 & 3 & 7 & 15\\
				\hline
			\end{tabular}
		\end{center}
	\end{table}
	
	Unboundedness of both sets pose possible problems for approximations of solution concepts. By Lemmata~\ref{lem:extreme-rays-convex},\ref{lem:extreme-rays-super}, both sets are bounded if and only if values of singletons are bounded. One way to solve this problem is to impose zero-normalisation to the extensions, which in combination with superadditivity implies monotonicity. We will proceed with a slightly more general case, where we restrict only to monotonicity. We denote by $\S_+$ and $\Co_+$ the sets of monotonic $\S$-extensions, respectively monotonic $\Co$-extensions. Finally, we discuss the question of $\S_+$-extendability and $\Co_+$-extendability. Of course, both extendabilities imply monotonicity of $(N,\K,v)$ together with superadditivity and convexity, respectively. The opposite direction is summarised in the following theorem.
	\begin{theorem}
		If a player-centered incomplete game $(N,\K,v)$ is monotonic, then it is $\S_+$-extendable. If it is also convex, then it is $\Co_+$-extendable.
	\end{theorem}
	\begin{proof}
		Consider $(N,v_1)$ defined in~\eqref{eq:min-max-vertices}. For $S,T \subseteq N$, condition
		\[
		v_1(S) + v_1(T) \leq v_1(S \cap T) + v_1(S \cup T)
		\]
		is equal to
		\begin{equation}\label{eq:C-S-extend1}
			v(S) + v(T) \leq v(S \cap T) + v(S \cup T)\text{ for }S,T \in \K,
		\end{equation}
		\begin{equation}\label{eq:C-S-extend2}
			v(S) \leq v(S \cup T)\text{ for }S \in \K, T \in \Kc\text{, and}
		\end{equation}
		\begin{equation}\label{eq:C-S-extend3}
			0 \leq 0\text{ for }S \in \Kc, T \in \Kc.
		\end{equation}
		Conditions~\eqref{eq:C-S-extend3} are always satisfied. Conditions~\eqref{eq:C-S-extend2} are satisfied if $(N,\K,v)$ is monotonic (means $(N,v_1)$ is $\S_+$-extension) and Conditions~\eqref{eq:C-S-extend1} are satisfied if $(N,\K,v)$ is convex.
	\end{proof}
	
	\subsection{Monotonic extensions}
	To initiate the analysis of monotonic extensions, we recall extensions $(N,v_0)$ and $(N,v_1)$ being defined as
	\begin{equation}
		v_0(S) = \begin{cases}
			v(S) & \text{if } S \in \K,\\
			v(S \cup i) & \text{if } S \notin \K,\\
		\end{cases}
		\hspace{0.2in} \text{and} \hspace{0.2in}
		v_1(S) = \begin{cases}
			v(S) & \text{if } S \in \K,\\
			0 & \text{if } S \notin \K.\\
		\end{cases}
	\end{equation}
	It is not difficult to see that if for $S,T \in \K$, $S \subseteq T$ it holds $v(S) \leq v(T)$, then both extensions are monotonic.
	\begin{theorem}
		A player-centered incomplete game $(N,\K,v)$ is $\M$-extendable if and only if it is monotonic.
	\end{theorem}
	\begin{proof}
		The forward implication is immediate. The opposite implication follows from monotonicity of extensions $(N,v_0)$ and $(N,v_1)$, which is satisfied when $(N,\K,v)$ is monotonic.
	\end{proof}
	From monotonicity, it must hold for every $\M$-extension $(N,w)$ and every $S \in \Kc$ that $0 \leq w(S) \leq w(S \cup i)$, or in terms of $(N,v_0)$, $(N,v_1)$ as
	\[
	v_0(S) \leq w(S) \leq v_1(S).
	\]
	Immediately, it follows that these games satisfy the definition of extreme games. 
	In the following theorem, we characterise all the vertices of
        $\M(v)$. The idea is as follows. Choose an ordering $\sigma$ of
        coalitions from $\Kc$ and denote them according to the ordering as
        $S_1,S_2,\dots,S_{\lvert\Kc\rvert}$. Now, in each step, decide if for
        coalition $S_i$, we set its worth to the minimal or the maximal possible
        value. The worth can be easily determined by taking either the maximum over all already defined worths of subsets or the minimum of the worths of already defined supersets. Intuitively, these games are extreme points because their worths are extremal. Formally, let $\sigma \in \Sigma_{\lvert \Kc\rvert}$ and let us denote $\L_\sigma(S) = \{A \in \Kc \mid \sigma(A) < \sigma(S)\} \cup \K$ the set of all predecessors of $S$ under $\sigma$. Now, for $x \in \{0,1\}^{\Kc}$, let $(N,\vsigmax)$ be a cooperative game defined for $S \in \K$ as $\vsigmax(S) \coloneqq v(S)$ and for $S \in \Kc$ as
	\begin{equation}\label{eq:vsigmax}
		\vsigmax(S) \coloneqq \begin{cases}
			\min\limits_{A \in \L_\sigma(S) , S \subsetneq A}\vsigmax(A) & \text{if } x_S = 1,\\
			\max\limits_{A \in \L_\sigma(S), A \subsetneq S}\vsigmax(A) & \text{if } x_S = 0.\\
		\end{cases}
	\end{equation}
	The permutation $\sigma$ represents the ordering described in the above
        process and the value $x_S$ represents the decision to either set the worth of $S$ to the minimal or the maximal possible value.
	
	\begin{theorem}
		For a $\M$-extendable player-centered $(N,\K,v)$, it holds $(N,\vsigmax)$ defined in~\eqref{eq:vsigmax} are extreme games of $\M(v)$.
	\end{theorem}
	\begin{proof}
		The game $(N,\vsigmax)$ is clearly an extension. By induction on $k$, we show that $(N,\L_\sigma(S_{k}),\vsigmax)$ is monotonic. Suppose $(N,\L_\sigma(S_{k}),v)$ is monotonic and $x_{S_k}=1$. For $T \in  \L_\sigma(S_{k}), S_k \subsetneq T$, we have
		\[
		\vsigmax(S_{k}) \leq \min\limits_{A \in \L_\sigma(S_{k}), S_k \subsetneq A}\vsigmax(A) \leq \vsigmax(T).
		\]
		For $T \in \L_\sigma(S_{k})$, $T \subsetneq S_k$, let $S^* \in \L_\sigma(S_{k})$ such that $\vsigmax(S_k)=\vsigmax(S^*)$. Since $S_k \subsetneq S^*$, we have $T \subsetneq S^*$. From monotonicity of $(N,\L_\sigma(S_{k}),v)$, it follows
		\[
		\vsigmax(T) \leq \vsigmax(S^*)=\vsigmax(S_k).
		\]
		
		As the case where $x_S = 0$ can be handled in a similar manner, monotonicity of $(N,\L_\sigma(S_{k+1}),\vsigmax)$ follows. For the initial step of the induction, $(N,\L_\sigma(S_1),\vsigmax) = (N,\K,v)$, which is monotonic by assumption.

		Finally, we show extremality. From the recursive construction of $(N,\vsigmax)$, for every $S \in \Kc$ we have $T \in \K$ (notice also $\emptyset \in \K$) such that
		\[
		\vsigmax(S) = v(T).
		\]
		This gives us $\lvert \Kc \rvert$ equalities, thus by Theorem~\ref{thm:extreme-point-char}, $(N,\vsigmax)$ is an extreme point of the set of $\M$-extensions.
	\end{proof}
	To prove these are the only extreme games, we need the following lemma.
	\begin{lemma}\label{lem:extreme-monotonic}
		Let $(N,e)$ be an extreme game of the set of $\M$-extensions of a $\M$-extendable player-centered incomplete game. Then for every $S \in \Kc$, there is either $k \in S$ such that
		\[
		e(S) = e(S \setminus k)
		\]
		or there is $j \in N \setminus S$ such that
		\[
		e(S) = e(S \cup j).
		\]
	\end{lemma}
	\begin{proof}
		Suppose there is neither $k$ nor $j$ satisfying the assertion above for some $S$. It means there is $\varepsilon > 0$ such that
		\[
		e(S \setminus k ) \leq e(S) + \varepsilon \leq e(S \cup j)
		\]
		and
		\[
		e(S \setminus k ) \leq e(S) - \varepsilon \leq e(S \cup j)
		\]
		for every $k$ and $j$. Define $(N,e^+)$ where $e^+(S)\coloneqq e(S) + \varepsilon$ and $e^+(T)\coloneqq e(T)$ otherwise. Similarly, we can define $(N,e^-)$. Both games are $\M$-extensions, which is in contradiction with Definition~\ref{def:extreme-point}.
	\end{proof}
	
	\begin{theorem}\label{thm:mono-set-description}
		For a $\M$-extendable player-centered incomplete game $(N,\K,v)$, it holds
		\[
		\M(v) = conv\left\{\vsigmax \mid \sigma \in \Sigma_{\lvert \Kc \rvert}, x \in \{0,1\}^{ \Kc }\right\}.
		\]
	\end{theorem}
	\begin{proof}
		Let $(N,e)$ be an extreme game of $\M(v)$. We construct $\sigma$ and $x$ such that $\vsigmax=e$. We proceed iteratively and in each step, we construct a set of coalitions $\L_i$. For the initial step, $\L_0 = \K$. Next, we collect all coalitions $S \in \Kc$ such that there is $S \setminus k \in \K$ or $S \cup j \in \K$ satisfying
		\[
		e(S) = e(S \setminus k) \text{ or } e(S) = e(S \cup j).
		\]
		To distinguish between two possible cases we denote by
		\[
		\L^+_1 = \{ S \in \Kc \mid \exists S \cup j \in \L_0: e(S) = e(S \cup j)\}
		\]
		and
		\[
		\L^-_1 = \{ S \in \Kc \setminus \L_1^+ \mid \exists S \setminus k \in \L_0: e(S) = e(S \setminus k)\}
		\]
		and $\L_1=\L^+_1 \cup \L^-_1$. Notice that it holds $\L^{-}_1 = \emptyset$. In each following step, we construct $\L_i$ in a similar manner, considering only coalitions $S$ not yet included by already constructed sets and coalitions $S \setminus k, S \cup j \in \L_{i-1}$. Formally,
		\[
		\L_i^+ = \{ S \in \Kc \setminus (\L_1\cup\dots\cup\L_{i-1}) \mid  \exists S \cup j \in \L_{i-1}, e(S) = e(S \cup j)\}
		\]
		and
		\[
		\L_i^- = \{ S \in \Kc \setminus (\L_1\cup \dots \L_{i-1}\cup \L_{i}^+)  \mid \exists S \setminus k \in \L_{i-1}, e(S) = e(S \setminus k)\}.
		\]
		We remark $\L_i^+ \cap \L_i^- = \emptyset$ and denote $\L = \cup_{i=0}^\infty\L_i$.
		
		First, we show that every $S \subseteq N$ is contained in $\L$.  For a contradiction, if $S \notin \L$, it must hold $S \in \Kc$. We construct sets $\L(S)$ and $\L_i(S)$ in a similar manner as $\L$ and $\L_i$ with the only distinction that $\L_0(S) = \{S\}$ and $\L_0=\K$. It holds for every $A,B \in \L(S)$ that $e(A) = e(S)$ and for every $X, Y \notin \L(S)$ such that $X \subsetneq A$ and $A \subsetneq Y$, it holds
		\begin{equation}\label{eq:6}
			e(X) < e(A)\hspace{0.2in}\text{and}\hspace{0.2in}e(A) < e(Y).
		\end{equation}
		Fix $\varepsilon > 0$ such that
		\begin{equation}\label{eq:7}
			e(X) < e(A) - \varepsilon\text{ and } e(A) + \varepsilon < e(Y).
		\end{equation}
		Such $\varepsilon$ exists, because for every $A \in \L(S)$, we have $X=\emptyset$ and $Y=N$ satisfying condition~\eqref{eq:7}. Define $(N,e^+)$, $(N,e^-)$ as
		\[
		(e^+)(T) \coloneqq \begin{cases}
			e(T) + \varepsilon & \text{if } T \in \L(S),\\
			e(T) &\text{if } T \notin \L(S),
		\end{cases}\hspace{0.2in}(e^-)(T) \coloneqq \begin{cases}
			e(T) - \varepsilon & \text{if } T \in \L(S),\\
			e(T) &\text{if } T \notin \L(S).
		\end{cases}
		\]
		These games are extensions because $\L(S) \subseteq \Kc$ and they are $\M$-extensions because they satisy~\eqref{eq:7} and because $(N,e)$ is monotonic. By definition of extreme points, $(N,e)$ is not an extreme game of $\M(v)$, a contradiction.
		
		Now, choose $\sigma \in \Sigma_{\lvert \Kc \rvert}$ satisfying
		\begin{equation}\label{eq:9}
			\sigma(S) < \sigma(T) \iff S \in \L_i, T \in \L_j\text{ where } 0 < i < j
		\end{equation}
		and construct $x \in \{0,1\}^{\Kc}$ such that
		\[
		x_S \coloneqq \begin{cases}
			1 &\text{if } S \in \bigcup_{i=1}^\infty\L_i^+,\\
			0 &\text{if } S \in \bigcup_{i=1}^\infty\L_i^-.\\
		\end{cases}
		\]
		Finally we show $\vsigmax = e$. For $S \in\K$, it holds
                $\vsigmax(S) = e(S) = v(S)$. Further, denote by
                $S_1,S_2,\dots,S_{\lvert \Kc \rvert}$ the ordering of coalitions
                in $\Kc$ according to $\sigma$. Now suppose that for $S_k$, it
                holds $\vsigmax(T)=e(T)$ for every $T \in \L(S)$. Further,
                suppose $S_k \in \L^+_i$ and $x_{S_k}=1$. This means that there exists $S_k \cup j \in \L_{i-1}$ such that
		\begin{equation}\label{eq:mono1}
			e(S_k)=e(S_k \cup j) = \vsigmax(S_k \cup j).
		\end{equation}
		and since $\L_{i-1} \subseteq \L(S_k)$, we have
		\begin{equation}\label{eq:mono2}
			\vsigmax(S_k) = \min\limits_{A \in \L_\sigma(S_k) , S_k \subsetneq A}\vsigmax(A) \leq \vsigmax(S_k \cup j).
		\end{equation}
		If $\vsigmax(S_k) < \vsigmax(S_k \cup j)$, there is $X  \in \L_\sigma(S_k)$
		such that $S_k \subsetneq X$ and
		\begin{equation}\label{eq:mono3}
			\vsigmax(S_k)=\vsigmax(X).
		\end{equation}
		As $\vsigmax(X)=e(X)$, we have
		\[
		e(X)=\vsigmax(S_k) < \vsigmax(S_k\cup j) = e(S_k),
		\]
		which contradicts $e(S_k) \leq e(X)$. Thus it must hold $\vsigmax(S_k)=e(S_k)$. For other cases, where $x_S=0$ and $S \in \L_{i}^-$, the analysis is similar.
	\end{proof}
	
	We remark multiple pairs $(\sigma,x)$ may describe the same extreme game. The ambiguity arises, e.g., from choosing permutation $\sigma$ according to~\eqref{eq:9}, however, this is not the only reason. %Nonetheless, this description allows us to work with approximations of solution concepts, because we have a way to work with all vertices of $\M$-extensions.
	
	\section{Approximations of solution concepts}
	Imagine that we want to determine a solution $\mathcal{S}(v^*)$ of a
        cooperative game $(N,v^*)$. However we have only partial knowledge about
        the game which is represented by an incomplete game $(N,\K,v)$. From its nature, we know that $(N,v^*)$ lies in $C \subseteq
        \Gamma^n$. The set of $C$-extensions of $(N,\K,v)$ represents the set of
        possible candidates for $(N,v^*)$, and $(N,v^*)$ is among these games. This means that by computing $\mathcal{S}(w)$ for every $C$-extension $(N,w)$, we also compute $\mathcal{S}(v^*)$. The only problem is that we cannot distinguish which one of the $C$-extensions is $(N,v^*)$. Nevertheless, by considering the union of $\mathcal{S}(w)$ for all $C$-extensions $(N,w)$, we are sure that $\mathcal{S}(v^*)$ is contained in this set. Similarly, if we consider the intersection of all solutions $\mathcal{S}(w)$, we have a set of payoff vectors which is guaranteed to be a subset of $\mathcal{S}(v^*)$. The idea is formally captured by the following definition.
	\begin{definition}\label{def:weak-solution-concept}
		Let $(N,\K,v)$ be a $C$-extendable incomplete game and $\mathcal{S}\colon C \to 2^{\Rn}$ a solution concept on $C \subseteq \Gamma^n$. Then by \emph{weak solution} $\cup\mathcal{S}(C,\K) \colon C(\K) \to 2^{\Rn}$, we mean
		\[
		\cup\mathcal{S}(C,\K)(v) \coloneqq \bigcup_{w \in C(v)}\mathcal{S}(w).
		\]
	\end{definition}
	We write $\cup\mathcal{S}(C)$ instead of $\cup\mathcal{S}(C,\K)(v)$ whenever $\K$ and $v$ is apparent from the context. In a similar way, we can define the \textit{strong solution} $\cap\mathcal{S}(C,\K)$ where the union is replaced by the intersection. It is clear from the definition that
	\begin{equation}
		\cap\mathcal{S}(C,\K)(v) \subseteq \mathcal{S}(v^*) \subseteq \cup\mathcal{S}(C,\K)(v)
	\end{equation}
	and the difference between the sets depends heavily on both $C$ and $\mathcal{K}$. If for example $\K=2^N$, all three sets coincide. If on the contrary $\K=\{\emptyset\}$, for most of the standard solution concepts, the relations become $\emptyset \subseteq \mathcal{S}(v^*) \subseteq \mathbb{R}^n$. Also, the more restrictive $C$ is, the less $C$-extensions are considered, thus the stronger the approximations are. The ultimate goal is to find a compromise between information provided by $(C,\K)$ and the strength of the approximations.
	
	So far, these approximations were studied for minimal incomplete games,
        a class of incomplete games with $\K = \{\{i\} \mid i \in N\} \cup
        \{\emptyset,N\}$. It was shown for most of the standard solution
        concepts that unless we restrict to
        $\P$-extensions, we get bounds of the form $\emptyset \subseteq
        \mathcal{S}(v^*) \subseteq \mathbb{I}(v)$. Interestingly, the core did not get better approximations even when restricted to $\P$-extensions. In this section, we analyse the weak and strong variants of the core, and the weak variant of the Shapley value and the $\tau$-value for player-centered incomplete games.
	
	\subsection{The core}
	In the case of the core, by considering $\cup\mathcal{C}$ and
        $\cap\mathcal{C}$, we are 
        able to get nontrivial bounds on
        $\mathcal{C}(v^*)$. More interestingly, these sets are
        attained for specific extensions, meaning that without further
        information, they are in a sense the best possible approximations.
	\begin{theorem}\label{thm:approx-core}
		Let $(N,\K,v)$ be an $i$-centered incomplete game. If
		\begin{enumerate}
			\item $(N,\K,v)$ is $\M$-extendable, then 
			\begin{enumerate}
				\item $\cup\Core(\M) = \Core(v_1)$,
			\end{enumerate}
			\item $(N,\K,v)$ is $\P$-extendable, then 
			\begin{enumerate}
				\item $\cap\Core(\P) = \Core(v_0)$, and 
				\item $\cup\Core(\P) = \Core(v_1)$.
			\end{enumerate}
		\end{enumerate}
	\end{theorem}
	\begin{proof}
		For an extension $(N,w)$ of $(N,\K,v)$, the core can be expressed as
		\[
			\Core(w) =\left\{x \in \mathbb{X}(v) \mid x_i \geq v(i), x(S\cup i) \geq v(S \cup i), x(S) \geq w(S), \forall S \in \Kc\right\}.
		\]
		Therefore, the only distinction between the cores of two different extensions is based on conditions
		\[
		x(S) \geq w(S)\text{ for } S \in \Kc.
		\]
		If we consider two extensions $(N,w_1)$, $(N,w_2)$ which differ only in a value of one fixed coalition $S \in \Kc$ such that $w_1(S) > w_2(S)$, then it holds $\Core(w_1) \subseteq \Core(w_2)$. This means that if there is a $C$-extension $(N,\overline{v})$ such that for every $S \in \Kc$,
		\[
		\overline{v}(S) \geq \max_{w \in C(v)}w(S)
		\]
		then $\Core(\overline{v}) = \cap \Core(C)$. Similarly, a $C$-extension $(N,\underline{v})$ such that for every $S \in \Kc$,
		\[
		\underline{v}(S) \leq \min_{w \in C(v)}w(S)
		\]
		would satisfy $\Core(\underline{v}) = \cup \Core(C)$. For both                sets of $\M$-extensions and $\P$-extensions such games exist (if the sets are nonempty) and it holds that $\overline{v} = v_0$ and $\underline{v}=v_1$.
	\end{proof}
	
	As it always holds that $\P(v) \subseteq \Co_+(v) \subseteq \S_+(v) \subseteq \M(v)$, Theorem~\ref{thm:approx-core} also shows that
	\begin{enumerate}
		\item $\cap\Core(\Co_+) = \Core(v_0)$ and $\cup\Core(\Co_+) = \Core(v_1)$ if $(N,\K,v)$ is $\Co_+$-extendable,
		\item $\cap\Core(\S_+) = \Core(v_0)$ and $\cup\Core(\S_+) = \Core(v_1)$ if $(N,\K,v)$ is $\S_+$-extendable.
	\end{enumerate}

	\subsection{The Shapley value}
	
	In this section, we investigate the weak Shapley value. We show that for player $i$, the range of its possible payoffs does not depend on the class of $C$-extensions, however, for the rest of the players, the range may differ. For a player-centered game $(N,\K,v)$, we define its Shapley value $\phi(v)$ as
	\[
	\phi_k(v) \coloneqq \sum_{S \in \K, k \notin S}\gamma_S\left(v(S \cup k)
        - v(S)\right), \quad k\in N
	\]
	where $\gamma_S \coloneqq \frac{s!(n-s-1)!}{n!}$. For an extension $(N,w)$ of $(N,\K,v)$ and every $k \in N$, we have\\ $\phi_k(w) = \sum_{S \subseteq N \setminus k}\gamma_S \left(w(S \cup k) - w(S)\right)$, which can be rewritten for $k\neq i$ as
	\[
	\sum_{S \in \K, k \notin S}\gamma_S\left(v(S \cup k) - v(S)\right) + \sum_{S \in \Kc, k \notin S}\gamma_S\left(w(S \cup k) - w(S)\right).
	\]
	We see that the value of the first sum is always equal to $\phi_k(v)$. Thus, it remains to determine the range of the second sum when we consider all $C$-extensions.
	
	\begin{theorem}\label{thm:shapley-mono-k}
		For a $\M$-extendable player center incomplete game $(N,\K,v)$, it holds for every $k \in N \setminus i$ that
		\[
		\cup\phi_k(\M) = \left[\phi_k(v),\phi_k(v) + \sum_{S \in \K, k \notin S}\gamma_{S}v(S \cup k)\right].
		\]
	\end{theorem}
	\begin{proof}
		
		To determine the range of $\cup\phi_k(\M)$, it is enough to determine the range of 
		\begin{equation}\label{eq:shapley-mono}
			\sum_{S \in \Kc, k \notin S}\gamma_S\left(w(S \cup k) - w(S)\right)
		\end{equation}
		across all $\M$-extensions $(N,w)$. From monotonicity of $w$, $w(S \cup k) - w(S) \geq 0$, therefore sum in~\eqref{eq:shapley-mono} is bounded from below by $0$. For $(N,v_0)$, we have $v_0(S\cup k) = v_0(S)$ for every $S \in \Kc$, thus \[
		\sum_{S \in \Kc, k \notin S}\gamma_S\left(v_0(S \cup k) - v_0(S) \right) = 0.\]
		It therefore yields the lower bound of $\cup\phi_k(\M)$. For the upper bound, we can rewrite~\eqref{eq:shapley-mono} as
		\begin{equation}\label{eq:shapley-mono2}
			\sum_{S \in \Kc, k \in S}\gamma_Sw(S) - \sum_{S \in \Kc, k \notin S}\gamma_Sw(S).
		\end{equation}
		The maximum across all $\M$-extensions of the sum
                in~\eqref{eq:shapley-mono} is smaller or equal to the maximum of
                the first sum minus the minimum of the second sum
                in~\eqref{eq:shapley-mono2}. Both the maximum and the minimum
                in~\eqref{eq:shapley-mono2} can be attained at the same time by
                choosing the extreme game $(N,\vsigmax)$ where 
		\[
		x_S = \begin{cases}
			1 & \text{if } k \in S,\\
			0 & \text{if } k \notin S.\\
		\end{cases}
		\]
		Notice that any permutation $\sigma \in \Sigma_{\lvert \Kc\rvert}$ yields the same game,
		\[
		\vsigmax(S) = \begin{cases}
			v(S \cup i) & \text{if } k \in S,\\
			0 & \text{if } k \notin S.\\
		\end{cases}
		\]
		For this extension, expression~\eqref{eq:shapley-mono2} is equal to
		\[
		\sum_{S \in \Kc,k \in S}\gamma_Sv(S \cup i) \text{, or equivalently} \sum_{S \in \K, k \notin S}\gamma_Sv(S \cup k),
		\]
		which concludes the proof.
	\end{proof}
	For player $i$, we can express
	\begin{equation}\label{eq:shapley-mono-i}
		\phi_i(w) = \sum_{S \subseteq N \setminus i}\gamma_S\left(w(S \cup i ) - w(S)\right) = \sum_{S \in \K}\gamma_{S \setminus i}v(S) - \sum_{S \notin \Kc}\gamma_Sw(S).
	\end{equation}
	Yet again, the value of the first sum is fixed and all that remains is to determine the range of the second sum.
 
	\begin{theorem}\label{thm:shapley-mono-i}
		For a $\M$-extendable player-centered incomplete game $(N,\K,v)$, it holds
		\begin{equation*}
			\cup\phi_i(\M)(v) = \left[\phi_i(v_0), \phi_i(v_1)\right].
		\end{equation*}
	\end{theorem}
	\begin{proof}
		To maximise the Shapley value of player $i$, we have to minimise sum
		\begin{equation}\label{eq:shapley-mono-i2}
			\sum_{S \notin \Kc}\gamma_Sw(S)
		\end{equation}
		across the set of all $\M$-extensions. From monotonicity, the sum is clearly non-negative, and for $(N,v_1)$, it is equal to $0$. To minimise~\eqref{eq:shapley-mono-i}, we  maximise~\eqref{eq:shapley-mono-i2}. This can be done by taking $(N,v_0)$, where $v_0(S)=v(S \cup i)$. It follows,
		\[
		\phi_i(v_0) = \sum_{S \subseteq N \setminus i}\gamma_S\left(v_0(S \cup i) - v_0(S)\right) = v(i).
		\]
	\end{proof}
	
	In the proof of Theorem~\ref{thm:shapley-mono-i}, the bounds of $\cup\phi_i(\M)$ were attained for games $(N,v_0)$ and $(N,v_1)$, which are positive if $(N,\K,v)$ is positive. Therefore, the same argument can be derived for $\cup\phi_i(\P)$.
	
	\begin{theorem}\label{thm:shapley-positive-i}
		For a $\P$-extendable player-centered incomplete game $(N,\K,v)$, it holds
		\[
		\cup\phi_i(\P) = \cup\phi_i(\M).
		\]
	\end{theorem}
	\begin{proof}
		Follows from the proof of Theorem~\ref{thm:shapley-mono-i}.
	\end{proof}
	
	It remains to determine $\cup\phi_k(\P)$. In this case, we cannot use the same argument as in proof of Theorem~\ref{thm:shapley-mono-k}, because the $\M$-extension attaining the upper bound is never positive. Therefore, we will proceed in a different manner, expressing the Shapley value of every $\P$-extensions, first. We denote by $(N,\valfa)$ a $\P$-extension defined as a convex combination of vertices of $\P(v)$, $\valfa = \sum_{x \in \{0,1\}^{\lvert \Kc\rvert}}\alpha_xv_x$ where $\alpha_x \geq 0 \forall x$ and $\sum_{x}\alpha_x=1$.
	
	\begin{lemma}\label{lem:shapley-pextension}
		Let $(N,\valfa)$ be a $\P$-extension of $i$-centered incomplete game $(N,\K,v)$. Then $\phi(\valfa)$ can be expressed as
		\[
		\phi_k(\valfa)=\sum_{S \in \Kc, k \in S}m^v(S\cup i) \left(\sum_{x: x_{S}=1}\frac{\alpha_x}{\lvert S \cup i\rvert} + \sum_{x:x_{S} = 0}\frac{\alpha_x}{\lvert S \rvert}\right)
		\]
		for every $k \in N \setminus i$ and
		\[
		\phi_i(\valfa) = m^v(i) + \sum_{S \in \K^c}m^v(S\cup i)\left(\sum_{x: x_{S}=1}\frac{\alpha_x}{\lvert S \cup i \rvert}\right).
		\]
	\end{lemma}
	\begin{proof}
		First, we express the Shapley value of an extreme game $(N,\vx)$. For the sake of brevity, in the rest of the proof, we substitute $\sum_x$ for $\sum_{x \in \{0,1\}^{\lvert\Kc\rvert}}$. It follows for every player $k \in N \setminus i$ that $\phi_k(\vx) = \sum_{S \subseteq N, k \in S} \frac{m^\vx(S)}{\lvert S \rvert}$ which further equals to
		\[
		\sum_{S \in \Kc, k \in S}\left(\frac{m^\vx(S\cup i)}{\lvert S \cup i \rvert} + \frac{m^\vx(S)}{\lvert S  \rvert}\right) = \sum_{S \in \Kc, k \in S}\frac{m^{v}(S\cup i)}{\lvert S \rvert + x_S}.
		\]
		Now for $\P$-extensions $(N,\valfa)$, we have $\phi(\valfa) = \sum_{x}\alpha_x\phi(\vx)$, which can be expressed as
		\begin{equation}\label{eq:lem-pshapley-mod1}
			\sum_{x}\alpha_x\sum_{S \in \Kc, k \in S}\frac{m^{v}(S\cup i)}{\lvert S \rvert + x_S} = \sum_{S \in \Kc, k \in S}\sum_{x}\frac{\alpha_xm^v(S \cup i)}{\lvert S \rvert + x_S}
		\end{equation}
		or
		\begin{equation}\label{eq:lem-pshapley-mod2}
		\sum_{S \in \Kc, k \in S} m^v(S \cup i) \sum_{x}\frac{\alpha_x}{\lvert S \rvert + x_S}
		\end{equation}
		For player $i$, we have
		\[
		\phi_i(\vx) = m^v(i) + \sum_{S \in \Kc}\frac{m^\vx(S\cup i)}{\lvert S \cup i \rvert} = \sum_{S \in \Kc}x_S\frac{m^\vx(S\cup i)}{\lvert S \cup i\rvert}.
		\]
		For $(N,\valfa)$, by similar modifications as in~\eqref{eq:lem-pshapley-mod1},~\eqref{eq:lem-pshapley-mod2}, we conclude
		\[
		\phi_i(\valfa) = m^v(i) + \sum_{S \in \K^c}m^v(S\cup i)\left(\sum_{x: x_{S}=1}\frac{\alpha_x}{\lvert S \cup i \rvert}\right).
		\]
	\end{proof}
	\begin{comment}
	    \begin{theorem}\label{thm:shapley-positive-k}
		Let $(N,\K,v)$ be a $\P$-extendable $i$-centered incomplete cooperative game. Then it holds for every $k \in N \setminus i$ that
		\[
		\cup\phi_k(\P) = \left[\sum_{S \in \Kc, k \in S}\frac{m^{v}(S \cup i)}{s+1},\sum_{S \in \Kc, k \in S}\frac{m^{v}(S \cup i)}{s}\right].
		\]
	\end{theorem}
	\begin{proof}
		From Lemma~\ref{lem:shapley-pextension}, we see that for player $k \in N \setminus i$, the payoff under the Shapley value varies based on the sum
		\[
		\sum_{x: x_S=1}\frac{\alpha_x}{\lvert S \rvert + 1} + \sum_{x:x_S = 0}\frac{\alpha_x}{\lvert S \rvert}.
		\]
		The larger the value of the expression, the larger the
                payoff. It follows that the minimum is attained for $x=1$ and the maximum for $x=0$.
	\end{proof}
	\end{comment}
	\begin{theorem}\label{thm:shapley-positive-k}
		Let $(N,\K,v)$ be a $\P$-extendable $i$-centered incomplete cooperative game. Then it holds for every $k \in N \setminus i$ that
		\[
		\cup\phi_k(\P) = \left[\phi_k(v_1),\phi_k(v_0)\right].
		\]
	\end{theorem}
	\begin{proof}
		From Lemma~\ref{lem:shapley-pextension}, we see that for player $k \in N \setminus i$, the payoff under the Shapley value varies based on the sum
		\[
		\sum_{x: x_S=1}\frac{\alpha_x}{\lvert S \rvert + 1} + \sum_{x:x_S = 0}\frac{\alpha_x}{\lvert S \rvert}.
		\]
		The larger the value of the expression, the larger the
                payoff. It follows that the minimum is attained for $x=1$ and the maximum for $x=0$, which corresponds to games $(N,v_1)$ and $(N,v_0)$, respectively.
	\end{proof}
	
	Theorems~\ref{thm:shapley-mono-i},\ref{thm:shapley-mono-k},\ref{thm:shapley-positive-i},\ref{thm:shapley-positive-k}, imply 
	\[
	\cup\phi_i(\M)=\cup\phi_i(\S_+)=\phi_i(\Co_+)=\cup\phi_i(\P)
	\]
	for $\P$-extendable player-centered $(N,\K,v)$. For $k \in N \setminus i$, only
	\[
	\underline{\cup\phi_k}(\M) = \underline{\cup\phi_k}(\S_+) = \underline{\cup\phi_k}(\Co_+) = \underline{\cup\phi_k}(\P)
	\]
	holds for $\P$-extendable player-centered incomplete game, but
	\[
	\overline{\cup\phi_k}(\M) = \overline{\cup\phi_k}(\P) +  \sum_{S \in \K, k \notin S}v(S).
	\]
	We do not show an exact upper bound for $\cup\phi_k(\S_+)$ and
        $\cup\phi_k(\Co_+)$, however, from numerical experiments it seems that both bounds may differ from $\overline{\cup\phi_k}(\M)$ and $\overline{\cup\phi_k}(\P)$.
	
	\subsection{The $\tau$-value}
	In this section, we restrict to \emph{zero-normalised} $\P$-extensions, i.e., $\P$-extensions $(N,w)$ satisfying $w(k)=0$ for every $k \in N$. We denote the set of zero-normalised positive extensions by $\P_0$-extensions. An $i$-centered incomplete game $(N,\K,v)$ is $\P_0$-extendable if and only if $v(i)=0$ and $(N,\K,v)$ is positive. Further, if we modify~\eqref{eq:positive-extensions-alternative}, any $\P_0(v)$-extension $(N,w)$ can be expressed as
	\begin{equation}\label{eq:0positive-extensions-alternative}
	w = \sum_{S \in \Kc} m^v(S \cup i) \left[\beta_Su_S + (1-\beta_S)u_{S \cup i}\right]
	\end{equation}
	where $\beta_{\{k\}}=0$ for every $k \in N \setminus i$ and $\beta_S \in \left[0,1\right]$.
	Now, formula~\eqref{eq:tau-value-mobious} for the $\tau$-value simplifies for every $\P_0$-extension $(N,w)$ to
	\begin{equation}\label{eq:tau-zero-normalised}
		\tau_k(w) = \frac{\sum_{S \subseteq N}m^w(S)}{\sum_{S \subseteq N}\lvert S \rvert m^w(S)}\sum_{S \subseteq N, k \in S}m^w(S),
	\end{equation}
	or since $v(N)=w(N)$, to
	\begin{equation}\label{eq:tau-positive-2}
		\tau_k(w) = \frac{v(N)}{\sum_{S \subseteq N}\lvert S \rvert m^w(S)}\sum_{S \subseteq N, k \in S}m^w(S).
	\end{equation}
	\begin{theorem}\label{thm:tau-positive-k}
		For a $\P_0$-extendable $i$-centered incomplete game $(N,\K,v)$ it holds for every $k \in N$, $k \neq i$, that
		\[
		\cup\tau_k(v) = \left[\tau_k(v_1),\tau_k(v_0)\right].
		\]
	\end{theorem}
	\begin{proof}
		Let $(N,w)$ be a $\P_0$-extension and $k \in N\setminus i$                fixed. From~\eqref{eq:0positive-extensions-alternative}, it
                follows 
		\[m^w(S) = \beta_S m^v(S\cup i)\hspace{0.2in}\text{ and }\hspace{0.2in}m^w(S \cup i) = (1-\beta_S) m^v(S \cup i)\]
		for every $S \in \Kc$. It follows that $m^{w}(S) + m^{w}(S\cup i) = m^v(S\cup i)$, thus
		\[
		\sum_{S \subseteq N, k \in S}m^{w}(S) = \sum_{S \in \Kc, k \in S}\left[m^{w}(S) + m^{w}(S \cup i)\right] = \sum_{S \in \K, k \in S}m^v(S).
		\]
		Further,
		\[
		\sum_{S \subseteq N} \lvert S \rvert \cdot m^w(S) = \sum_{S \in \Kc} \left[s \cdot m^w(S) + (s+1)\cdot m^w(S\cup i)\right]
		\]
		or equivalently,
		\begin{equation}\label{eq:sum-positive}
			\sum_{S \in \Kc} m^v(S\cup i)\left(s+1-\beta_S\right).
		\end{equation}
		We see that $\tau_k(w)$ varies over the set of $\P$-extensions based on the value of the sum in~\eqref{eq:sum-positive}.
		
		This sum is maximal if and only if $\beta_S=0$ for every $S \in \Kc$ and minimal if and only if $\beta_S=1$ for every $S \in \Kc$. Then $\tau_k$ is minimal if and only if the sum is maximal and vice versa. A $\P$-extension, where $\beta_S=0$ for every $S \in \Kc$ (minimal $\tau_k$) equals $(N,v_1)$ and if $\beta_S=1$ for every $S \in \Kc$ (maximal $\tau_k$), it is equal to $(N,v_0)$.
	\end{proof}
	
	\begin{theorem}
		For a $\P_0$-extendable $i$-centered incomplete game $(N,\K,v)$ it holds
		\[
		\cup\tau_i(v) = \left[\tau_i(v_0),\tau_i(v_1)\right].
		\]
	\end{theorem}
	\begin{proof}
		We can rewrite~\eqref{eq:tau-positive-2} as
		\[
		\tau_i(w) = v(N) \frac{\sum_{S \subseteq N, i \in S}m^w(S)}{\sum_{S \subseteq N}\lvert S \rvert m^w(S)}.
		\]
		From~\eqref{eq:0positive-extensions-alternative}, it follows
		\[
		\sum_{S \subseteq N, i \in S}m^w(S) = \sum_{S \in \Kc}(1-\beta_S)m^v(S \cup i)
		\]
		and we already showed in the proof of Theorem~\ref{thm:tau-positive-k} that
		\[
		\sum_{S \subseteq N} \lvert S \rvert m^w(S) = \sum_{S \in \Kc} m^v(S\cup i)\left(s+1-\beta_S\right).
		\]
		By substituting $a=\sum_{S \in \Kc}(1-\beta_S)m^v(S \cup i)$ and $b=\sum_{S \in \Kc}s \cdot m^v(S \cup i)$, we get
		\[
		\tau_i(w)=\frac{a}{a+b}.
		\]
		Now, for another $\P_0$-extension $(N,u)$, denote by $\beta^u_S$ for every $S \in \Kc$ its coefficients with respect to~\eqref{eq:0positive-extensions-alternative}. Further, fix $T \in \Kc$ and choose $(N,u)$ such that $\beta_S=\beta^u_S$ for every $S \in \Kc \setminus \{T\}$ and set $\beta^u_T \neq \beta_T$. If we denote by $c = (\beta_T-\beta^u_T)m^v(T \cup i)$, we can express
		\[
		\tau_i(u) = \frac{a+c}{a+b+c}.
		\]
		Now $\tau_i(w) \leq \tau_i(u)$ if and only if
		\[
		\frac{a}{a+b} \leq \frac{a+c}{a+b+c}.
		\]
		As
		\[
		\frac{a}{a+b} = \frac{a^2 + ab + ac}{(a+b)(a+b+c)}\hspace{0.2in}\text{and}\hspace{0.2in}\frac{a+c}{a+b+c} = \frac{a^2 + ab + ac + bc}{(a+b)(a+b+c)},
		\]
		we see that $\tau_i(w) \leq \tau_i(u)$ if and only if $c > 0$
                which means $\beta_T > \beta^u_T$. This implies that the maximal
                $\tau_i$ is attained for $\beta_S=0$ for every $S \in \Kc$,
                which yields game $(N,v_1)$, while the minimal $\tau_i$ is
                reached for $\beta_S=1$ for every $S \in \Kc$, which yields game
                $(N,v_0)$. 
	\end{proof}
	
	\section{Conclusions}
        Based on our methods introduced in~\cite{Cerny2022}, we derived approximations of standard solution concepts of cooperative games with partial information. We focused on the case where known information is centered around one specific player and we investigated different sets of $C$-extensions and analysed approximations of the core, the Shapley value and the $\tau$-value. We showed there are two extensions $(N,v_0)$, $(N,v_1)$ (defined in~\eqref{eq:min-max-vertices}) playing the key role in the description of the approximations. Intuitively, they reflect contribution of the player around which the incomplete game is centered.
        There are two structures of $\K$, which are connected to the case studied in this paper, which we would like to address in the near future:
         \begin{enumerate}
            \item $\K \subseteq 2^N: S,T \in \K \implies S \cup T \in \K$,
            \item $\K = \{ S \subseteq N \mid i \notin S\}$ for a fixed $i \in N$.
         \end{enumerate}
         The first structure is a generalisation of $\K_i$ studied in this text and was studied under the scope of restricted games (see~\cite{Brink2023}). The second structure represents a game on $N \setminus i$. The question of $C$-extendability is then equivalent to adding player $i$ to the game in such a way that the game remains in $C$. We believe there might be close relations between games with this structure and games with $\K_i$.	
 
        The next big step regarding the approximations is to derive tools to further
analyse the strength of the approximations. So far, we considered only the
inclusion as a measure that one approximation is better than another. For
example, if there is a $C$-extension for which the solution concept is equal
exactly to the weak solution (see proof of Theorem 28), the approximation is
clearly the best possible we can get if we consider that the underlying game is from
$C$. But consider weak solutions of one-point solution concepts. Clearly, the smaller
the weak solution (in volume, range of values, ...), the better approximation
we are getting. We also want to address these questions in near future.

\section*{Acknowledgement}
 The first author was supported by the Charles University Grant Agency (GAUK 341721) and by the grant P403-22-11117S of the Czech Science Foundation.

%\bibliographystyle{alpha}
%\bibliography{sample}

\end{document}